\documentclass[11pt]{article}
\usepackage[letterpaper,hmargin=1in,vmargin=1.25in]{geometry}
\newcommand{\keywords}[1]{\textbf{Keywords.} #1}
\usepackage{amsthm}
\theoremstyle{definition}
\newtheorem{definition}{Definition}
\newtheorem{remark}{Remark}
\theoremstyle{theorem}
\newtheorem{proposition}{Proposition}

\usepackage{mathptmx,ifpdf}
\usepackage[utf8]{inputenc}
\usepackage[english]{babel}
\usepackage{amsmath,amsfonts,mathtools,epstopdf,stmaryrd}
\usepackage{color,xspace,xcolor,graphicx,paralist,wrapfig,balance}
\usepackage{relsize,alltt,multirow,booktabs,paralist,fancyvrb}
\ifpdf
\usepackage[pdftex,
            pdfauthor={Pablo Rauzy, Sylvain Guilley},
						colorlinks=false,pdfborder={0 0 0},
            pdftitle={A Formal Proof of Countermeasures Against Fault Injection Attacks on CRT-RSA}
            ]{hyperref}
\usepackage{xmpincl}
\includexmp{CC_Attribution_4.0_International}
\else
\usepackage{hyperref}
\fi

\usepackage[ruled,linesnumbered]{algorithm2e}
\newcommand{\algorithmConfig}{
  \DontPrintSemicolon
  \SetKw{Or}{or}
  \SetKwInOut{Input}{Input}
  \SetKwInOut{Output}{Output}
}
\makeatletter
\newcommand{\removelatexerror}{\let\@latex@error\@gobble}
\makeatother

\newcommand{\finja}{\textsf{finja}\xspace}

\newcommand{\Z}{\mathbb{Z}}
\newcommand{\ie}{\textit{i.e.}}
\newcommand{\eg}{\textit{e.g.}}
\newcommand{\etal}{\textit{et al.}\xspace}
\newcommand{\etc}{\textit{etc.}\xspace}

\title{A Formal Proof of Countermeasures Against Fault Injection Attacks on CRT-RSA}

\author{Pablo Rauzy --- Sylvain Guilley\\
Institut Mines-Télécom ; Télécom ParisTech ; CNRS LTCI\\
\{{\it firstname}.{\it lastname}\}@telecom-paristech.fr}

\begin{document}

\maketitle

\begin{abstract}
In this article, we describe a methodology that aims at either breaking or proving the security of CRT-RSA implementations against fault injection attacks.
In the specific case-study of the BellCoRe attack,
our work bridges a gap between formal proofs and implementation-level attacks.
We apply our results to three implementations of CRT-RSA,
namely the unprotected one, that of Shamir, and that of Aumüller \etal
Our findings are that many attacks are possible on both the unprotected and the Shamir implementations,
while the implementation of Aumüller \etal is resistant to all single-fault attacks.
It is also resistant to double-fault attacks if we consider the less powerful threat-model of its authors.

\keywords{RSA (\textit{Rivest, Shamir, Adleman}~\cite{DBLP:journals/cacm/RivestSA78}) \and
CRT (\textit{Chinese Remainder Theorem}) \and
fault injection \and
BellCoRe (\textit{Bell Communications Research}) attack \and
formal proof \and
OCaml}
\end{abstract}

\section{Introduction}

It is known since 1997 that injecting faults during the computation of CRT-RSA could yield to malformed signatures that expose the prime factors ($p$ and $q$) of the public modulus ($N=p \cdot q$).
Notwithstanding, computing without the fourfold acceleration conveyed by the CRT is definitely not an option in practical applications.
Therefore, many countermeasures have appeared that consist in step-wise internal checks during the CRT computation.
To our best knowledge, none of these countermeasures have been proven formally.
Thus without surprise, some of them have been broken, and then patched.
The current state-of-the-art in computing CRT-RSA without exposing $p$ and $q$ relies thus on algorithms that have been carefully scrutinized by cryptographers.
Nonetheless, neither the hypotheses of the fault attack nor the security itself have been unambiguously modeled.

This is the purpose of this paper.
The difficulties are \emph{a priori} multiple:
in fault injection attacks, the attacker has an extremely high power because he can fault any variable.
Traditional approaches thus seem to fall short in handling this problem.
Indeed, there are two canonical methodologies:
\emph{formal} and \emph{computational} proofs.
Formal proofs (\eg, in the so-called Dolev-Yao model) do not capture the requirement for faults to preserve some information about one of the two moduli; indeed, it considers the RSA as a black-box with a key pair.
Computational proofs are way too complicated (in terms of computational complexity) since the handled numbers are typically $2,048$~bit long.

The state-of-the-art contains one reference related to the formal proof of a CRT-RSA implementation: it is the work of Christofi, Chetali, Goubin and Vigilant~\cite{JCEN-Christofi13}.
For tractability purposes, the proof is conducted on reduced versions of the algorithms parameters.
One fault model is chosen authoritatively (the zeroization of a complete intermediate data), which is a strong assumption.
In addition, the verification is conducted on a pseudo-code, hence concerns about its portability after its compilation into machine-level code.
Another reference related to formal proofs and fault injection attacks is the work of Guo, Mukhopadhyay, and Karri.
In~\cite{cryptoeprint:2012:552}, they explicit an AES implementation that is provably protected against differential fault analyses~\cite{biham97differential}.
The approach is purely combinational, because the faults propagation in AES concerns $32$-bit words called columns;
consequently, all fatal faults (and thus all innocuous faults) can be enumerated.

\paragraph{Contributions.}
Our contribution is to reach a full fault coverage of the CRT-RSA algorithm, thereby keeping the proof valid even if the code is transformed (\eg, compiled or partitioned in software/hardware).
To this end we developed a tool called \finja\footnote{\url{http://pablo.rauzy.name/sensi/finja.html}} based on symbolic computation in the framework of modular arithmetic,
which enables formal an\-al\-ysis of CRT-RSA and its countermeasures against fault injection attacks.
We apply our methods on three implementations: the unprotected one, the one protected by Shamir's countermeasure, and the one protected by Aumüller \etal's countermeasure.
We find many possible fault injections that enable a BellCoRe attack on the unprotected implementation of the CRT-RSA computation, as well as on the one protected by Shamir's countermeasure.
We formally prove the security of the Aumüller \etal's countermeasure against the BellCoRe attack,
under a fault model that considers \emph{permanent faults} (in memory) and \emph{transient faults} (one-time faults, even on copies of the secret key parts),
with or without forcing at zero, and with possibly faults at various locations.

\paragraph{Organization of the paper.}
We recall the CRT-RSA cryptosystem and the BellCoRe attack in Sec.~\ref{sec-bellcore};
still from an historical perspective, we explain how the CRT-RSA implementation has been amended to withstand more or less efficiently the BellCoRe attack.
Then, in Sec.~\ref{sec-approach}, we define our approach.
Sec.~\ref{sec-unprotected}, Sec.~\ref{sec-shamir}, and Sec.~\ref{sec-aumuller} are case studies using the methods developed in Sec.~\ref{sec-approach} of respectively an unprotected version of the CRT-RSA computation, a version protected by Shamir's countermeasure, and a version protected by Aumüller \etal's countermeasure.
Conclusions and perspectives are in Sec.~\ref{sec-perspectives}.

\section{CRT-RSA and the BellCoRe Attack}
\label{sec-bellcore}

This section recaps known results about fault injection attacks on CRT-RSA (see also \cite{koc_rsa} and \cite[Chap. 3]{Intro_HOST}).
Its purpose is to settle the notions and the associated notations that will be used in the later sections (that contain novel contributions).

\subsection{CRT-RSA}
\label{sub-CRTRSA}

RSA is both an \emph{encryption} and a \emph{signature} scheme.
It relies on the identity that for all message $0\leq m<N$,
$(m^d)^e \equiv m \mod N$, where $d \equiv e^{-1} \mod \varphi(N)$, by Euler's theorem.
In this equation, $\varphi$ is Euler's totient function, equal to $\varphi(N)=(p-1) \cdot (q-1)$ when $N=p \cdot q$ is a composite number, product of two primes $p$ and $q$.
For example, if Alice generates the signature $S=m^d \mod N$,
then Bob can verify it by computing $S^e \mod N$, which must be equal to $m$ unless Alice is pretending to know $d$ although she does not.
Therefore the pair $(N,d)$ is called the private key, while the pair $(N,e)$ is called the public key.
In this paper, we are not concerned about the key generation step of RSA,
and simply assume that $d$ is an unknown number in $\llbracket 1, \varphi(N)=(p-1) \cdot (q-1)\llbracket$.
Actually, $d$ can also be chosen equal to the smallest value $e^{-1} \mod \lambda(n)$,
where $\lambda(n) = \frac{(p-1) \cdot (q-1)}{\gcd(p-1, q-1)}$ is the Carmichael function.
The computation of $m^d \mod N$ can be speeded-up by a factor four by using the Chinese Remainder Theorem (CRT).
Indeed, figures modulo $p$ and $q$ are twice as short as those modulo $N$.
For example, for $2,048$~bit RSA, $p$ and $q$ are $1,024$~bit long.
The CRT-RSA consists in computing $S_p = m^d \mod p$ and $S_q = m^d \mod q$,
which can be recombined into $S$ with a limited overhead.
Due to the little Fermat theorem (special case of the Euler theorem when the modulus is a prime),
$S_p = (m \mod p)^{d \mod (p-1)} \mod p$.
This means that in the computation of $S_p$, the processed data have $1,024$~bit,
and the exponent itself has $1,024$~bits (instead of $2,048$~bits).
Thus the multiplication is four times faster and the exponentiation eight times faster.
However, as there are two such exponentiations (modulo $p$ and $q$), the overall CRT-RSA is roughly speaking four times faster than RSA computed modulo $N$.

This acceleration justifies that CRT-RSA is always used if the factorization of $N$ as $p \cdot q$ is known.
In CRT-RSA, the private key is a more rich structure than simply $(N,d)$:
it is actually comprised of the $5$-tuple $(p, q, d_p, d_q, i_q)$, where:
\begin{compactitem}
\item $d_p \doteq d \mod (p-1)$,
\item $d_q \doteq d \mod (q-1)$,
\item $i_q \doteq q^{-1} \mod p$.
\end{compactitem}
The unprotected CRT-RSA algorithm is presented in Alg.~\ref{alg-crt-rsa-unprotected}.
It takes advantage of the CRT recombination proposed by Garner in~\cite{DBLP:journals/ac/Garner65}. 
It is straightforward to check that the signature computed at line~\ref{alg-crt-rsa-unprotected-S}
belongs to $\llbracket 0, p \cdot q -1 \rrbracket$.
Consequently, no reduction modulo $N$ is necessary before returning $S$.

\begin{algorithm}
\algorithmConfig
\Input{Message $m$, key $(p, q, d_p, d_q, i_q)$}
\Output{Signature $m^d \mod N$}
\BlankLine
$S_p = m^{d_p} \mod p$ \tcc*[r]{Signature modulo $p$}
$S_q = m^{d_q} \mod q$ \tcc*[r]{Signature modulo $q$}
$S = S_q + q \cdot (i_q \cdot (S_p-S_q) \mod p)$ \tcc*[r]{Recombination} \label{alg-crt-rsa-unprotected-S}
\Return $S$ \;

\caption{Unprotected CRT-RSA}
\label{alg-crt-rsa-unprotected}
\end{algorithm}

\subsection{BellCoRe Attack on CRT-RSA}

In 1997, an dreadful remark has been made by Boneh, DeMillo and Lipton~\cite{boneh-fault}, three staff members of BellCoRe:
Alg.~\ref{alg-crt-rsa-unprotected} could reveal the secret primes $p$ and $q$ if the computation is faulted, even in a very random way.
The attack can be expressed as the following proposition.
\begin{proposition}[Orignal BellCoRe attack]
If the intermediate variable $S_p$ (resp. $S_q$) is returned faulted as $\widehat{S_p}$ (resp. $\widehat{S_q}$)\footnote{In other papers related to faults, the faulted variables (such as $X$) are noted either with a star ($X^*$) or a tilde ($\tilde{X}$); in this paper, we use a hat, as it can stretch, hence cover the adequate portion of the variable. For instance, it allows to make an unambiguous difference between a faulted data raised at some power and a fault on a data raised at a given power (contrast $\widehat{X}^e$ with $\widehat{X^e}$).},
then the attacker gets an erroneous signature $\widehat{S}$,
and is able to recover $p$ (resp. $q$) as $\gcd(N, S-\widehat{S})$
(\emph{with an overwhelming probability}).
\label{pro-bellcore2faults}
\end{proposition}
\begin{proof}
For all integer $x$, $\gcd(N, x)$ can only take $4$ values:
\begin{compactitem}
\item $1$, if $N$ and $x$ are coprime,
\item $p$, if $x$ is a multiple of $p$ but not of $q$,
\item $q$, if $x$ is a multiple of $q$ but not of $p$,
\item $N$, if $x$ is a multiple of both $p$ and $q$, \ie, of $N$.
\end{compactitem}
\smallskip
In Alg.~\ref{alg-crt-rsa-unprotected}, if $S_p$ is faulted (\ie, replaced by $\widehat{S_p} \neq S_p$),
then \\
$S-\widehat{S} = q \cdot (( i_q \cdot (S_p-S_q) \mod p) - ( i_q \cdot (\widehat{S_p}-S_q) \mod p))$,\\
and thus $\gcd(N, S-\widehat{S}) = q$.\\
If $S_q$ is faulted (\ie, replaced by $\widehat{S_q} \neq S_q$),
then \\
$S-\widehat{S} \equiv (S_q - \widehat{S_q}) - ( q \mod p) \cdot i_q \cdot (S_q - \widehat{S_q}) \equiv 0 \mod p$ because $(q \mod p) \cdot i_q \equiv 1 \mod p$.
Thus $S-\widehat{S}$ is a multiple of $p$.
Additionally, $S-\widehat{S}$ is not a multiple of $q$.\\
So $\gcd(N, S-\widehat{S})=p$.

In both cases, the greatest common divisor could yield $N$.
However, $(S-\widehat{S})/q$ in the first case (resp. $(S-\widehat{S})/p$ in the second case)
is very unlikely to be a multiple of $p$ (resp. $q$).
Indeed, if the random fault is uniformly distributed,
the probability that $\gcd(N, S-\widehat{S})$ is equal to $p$ (resp. $q$) is negligible%
\footnote{If it nonetheless happens that $\gcd(N, S-\widehat{S})=N$,
then the attacker can simply retry another fault injection, for which the probability that $\gcd(N, S-\widehat{S}) \in\{p,q\}$ increases.}.
\qed\end{proof}

This version of the BellCoRe attack requires that two identical messages with the same key can be signed;
indeed, one signature that yields the genuine $S$ and an other that is perturbed and thus returns $\widehat{S}$ are needed.
Little later, the BellCoRe attack has been improved by Joye, Lenstra and Quisquater~\cite{DBLP:journals/joc/JoyeLQ99}.
This time, the attacker can recover $p$ or $q$ with one only faulty signature,
provided the input $m$ of RSA is known.

\begin{proposition}[One faulty signature BellCoRe attack]
If the intermediate variable $S_p$ (resp. $S_q$) is returned faulted as $\widehat{S_p}$ (resp. $\widehat{S_q}$),
then the attacker gets an erroneous signature $\widehat{S}$,
and is able to recover $p$ (resp. $q$) as $\gcd(N, m-\widehat{S}^e)$
(\emph{with an overwhelming probability}).
\label{pro-bellcore1fault}
\end{proposition}
\begin{proof}
By proposition~\ref{pro-bellcore2faults}, if a fault occurs during the computation of $S_p$,
then $\gcd(N, S-\widehat{S})=q$ (\emph{most likely}).
This means that:
\begin{compactitem}
\item $S \not\equiv \widehat{S} \mod p$, and thus $S^e \not\equiv \widehat{S}^e \mod p$
(\emph{indeed, if the congruence was true, we would have $e|p-1$, which is very unlikely});
\item $S \equiv \widehat{S} \mod q$, and thus $S^e \equiv \widehat{S}^e \mod q$;
\end{compactitem}
As $S^e \equiv m \mod N$, this proves the result.
A symmetrical reasoning can be done if the fault occurs during the computation of $S_q$.
\qed\end{proof}

\subsection{Protection of CRT-RSA Against the BellCoRe Attack}

Several protections against the BellCoRe attack have been proposed.
A non-exhaustive list is given below, and then, the most salient features of these countermeasures are described:

\begin{itemize}
\item Obvious countermeasures: no CRT optimization, or with signature verification;
\item Shamir~\cite{shamir-patent-rsa-crt};
\item Aumüller \etal~\cite{DBLP:conf/ches/AumullerBFHS02};
\item Vigilant, original~\cite{DBLP:conf/ches/Vigilant08} and with some corrections by Coron \etal~\cite{DBLP:conf/fdtc/CoronGMPV10};
\item Kim \etal~\cite{Kim:2011:ECA:2010601.2010865}.
\end{itemize}

\subsubsection{Obvious Countermeasures}

Fault attacks on RSA can be thwarted simply by refraining from implementing the CRT.

If this is not affordable, then the signature can be verified before being outputted.
If $S=m^d \mod N$ is the signature, this straightforward countermeasure consists in testing $S^e \stackrel{?}{\equiv} m \mod N$.
Such protection is efficient in practice, but is criticized for three reasons.
First of all, it requires an access to $e$, which is not always present in the secret key structure,
as in the $5$-tuple example given in Sec.~\ref{sub-CRTRSA}.
Nonetheless, we attract the author's attention on paper~\cite{DBLP:conf/fdtc/Joye09} for a clever embedding of $e$ into [the representation of] $d$. 
Second, the performances are incurred by the extra exponentiation needed for the verification.
In some applications, the public exponent can be chosen small (for instance $e$ can be equal to a number such as $3$, $17$ or $65537$),
and then $d$ is computed as $e^{-1} \mod \lambda(N)$ using the extended Euclidean algorithm or better alternatives~\cite{DBLP:conf/ches/JoyeP03}.
But in general, $e$ is a number possibly as large as $d$ (both are as large as $N$),
thus the obvious countermeasure doubles the computation time (which is really non-negligible, despite the CRT fourfold acceleration).
Third, this protection is not immune to a fault injection that would target the comparison.
Overall, this explains why other countermeasures have been devised.

\subsubsection{Shamir}

The CRT-RSA algorithm of Shamir builds on top of the CRT and introduces, in addition to the two primes $p$ and $q$, a third factor $r$.
This factor $r$ is random%
\footnote{The authors notice that in Shamir's countermeasure, $r$ is \emph{a priori} not a secret, hence can be static and safely divulged.}
and small (less than $64$~bit long),
and thus co-prime with $p$ and $q$.
The computations are carried out modulo $p' = p \cdot r$ (resp. modulo $q' = q \cdot r$),
which allows for a retrieval of the intended results by reducing them modulo $p$ (resp. modulo $q$),
and for a verification by a reduction modulo $r$.
Alg.~\ref{alg-crt-rsa-shamir} describes one version of Shamir's countermeasure.
This algorithm is aware of possible fault injections, and thus can raise an \emph{exception} if an incoherence is detected.
In this case, the output is not the (purported faulted) signature, but a specific message ``$\mathsf{error}$''.

\begin{algorithm}
\algorithmConfig
\Input{Message $m$, key $(p, q, d, i_q)$,\\
$32$-bit random prime $r$}
\Output{Signature $m^d \mod N$,\\
or $\mathsf{error}$ if some fault injection has been detected.}
\BlankLine
$p' = p \cdot r$ \;
$d_p=d \mod (p-1) \cdot (r-1)$ \;
$S'_p = m^{d_p} \mod p'$ \tcp*[r]{Signature modulo $p'$}
\BlankLine
$q' = q \cdot r$ \;
$d_q=d \mod (q-1) \cdot (r-1)$ \;
$S'_q = m^{d_q} \mod q'$ \tcp*[r]{Signature modulo $q'$}
\BlankLine
$S_p = S'_p \mod p$ \; \label{alg-crt-rsa-shamir-sp}
$S_q = S'_q \mod q$ \;
\BlankLine
\tcp*[l]{\label{alg-crt-rsa-shamir-recombination}Same as in line~\ref{alg-crt-rsa-unprotected-S} of Alg.~\ref{alg-crt-rsa-unprotected}}
$S = S_q + q \cdot (i_q \cdot (S_p-S_q) \mod p)$ \;
\BlankLine
\eIf{$S'_p \not\equiv S'_q \mod r$}{ \label{alg-crt-rsa-shamir-test}
	\Return $\mathsf{error}$ \;
}{
	\Return $S$ \;
}

\caption{Shamir CRT-RSA}
\label{alg-crt-rsa-shamir}
\end{algorithm}

\subsubsection{Aumüller}

The CRT-RSA algorithm of Aumüller \etal is a variation of that of Shamir,
that is primarily intended to fix two shortcomings.
First it removes the need for $d$ in the signature process,
and second, it also checks the recombination step.
The countermeasure, given in Alg.~\ref{alg-crt-rsa-aumuller},
introduces, in addition to $p$ and $q$, a third prime $t$.
The computations are done modulo $p' = p \cdot t$ (resp. modulo $q' = q \cdot t$),
which allows for a retrieval of the intended results by reducing them modulo $p$ (resp. modulo $q$),
and for a verification by a reduction modulo $t$.
However, the verification is more subtle than for the case of Shamir.
In Shamir's CRT-RSA (Alg.~\ref{alg-crt-rsa-shamir}),
the verification is \emph{symmetrical}, in that the computations modulo $p \cdot r$ and $q \cdot r$ operate on the same object, namely $m^d$.
In Aumüller \etal's CRT-RSA (Alg.~\ref{alg-crt-rsa-aumuller}),
the verification is \emph{asymmetrical}, since the computations modulo $p \cdot t$ and $q \cdot t$ operate on two different objects, namely $m^{d_p \mod (t-1)}$ and $m^{d_q \mod (t-1)}$.
The verification consists in an identity that resembles that of ElGamal for instance: is $(m^{d_p \mod (t-1)})^{d_q \mod (t-1)}$ equivalent to $(m^{d_q \mod (t-1)})^{d_p \mod (t-1)}$ modulo $t$?
Specifically, if we note $S'_p$ the signature modulo $p'$,
then $S_p = S \mod p$ is equal to $S'_p \mod p$.
Furthermore, let us denote
\begin{compactitem}
\item $S_{pt} = S'_p \mod t$,
\item $S_{qt} = S'_q \mod t$,
\item $d_{pt} = d_p \mod (t-1)$, and
\item $d_{qt} = d_q \mod (t-1)$.
\end{compactitem}
It can be verified that those figures satisfy the identity:
$S_{pt}^{d_{qt}} \equiv S_{qt}^{d_{pt}} \mod t$,
because both terms are equal to $m^{d_{pt} \cdot d_{qt}} \mod t$.
The prime $t$ is referred to as a security parameter, as the probability to pass the test (at line~\ref{alg-crt-rsa-aumuller-ElGamal} of Alg.~\ref{alg-crt-rsa-aumuller}) is equal to $1/t$ (\ie, about $2^{-32}$),
assuming a uniform distribution of the faults.
Indeed, this is the probability to find a large number that, once reduced modulo $t$, matches a predefined value.

\begin{algorithm}
\algorithmConfig
\Input{Message $m$, key $(p, q, d_p, d_q, i_q)$,\\
$32$-bit random prime $t$}
\Output{Signature $m^d \mod N$,\\
or $\mathsf{error}$ if some fault injection has been detected.}
\BlankLine
$p'=p \cdot t$ \;
$d'_p=d_p + \mathsf{random}_1 \cdot (p-1)$ \tcp*[r]{Against SPA, not fault attacks} \label{alg-crt-rsa-aumuller-line_dpa_p}
$S'_p = m^{d'_p} \mod p'$ \tcp*[r]{Signature modulo $p'$}
\If{$(p' \mod p \neq 0)$ \Or $(d'_p \not\equiv d_p \mod (p-1))$}{ \label{alg-crt-rsa-aumuller-test1}
	\Return $\mathsf{error}$
}
\BlankLine
$q'=q \cdot t$ \;
$d'_q=d_q + \mathsf{random}_2 \cdot (q-1)$ \tcp*[r]{Against SPA, not fault attacks} \label{alg-crt-rsa-aumuller-line_dpa_q}
$S'_q = m^{d'_q} \mod q'$ \tcp*[r]{Signature modulo $q'$}
\If{$(q' \mod q \neq 0)$ \Or $(d'_q \not\equiv d_q \mod (q-1))$}{ \label{alg-crt-rsa-aumuller-test2}
	\Return $\mathsf{error}$ \;
}
\BlankLine
$S_p = S'_p \mod p$ \;
$S_q = S'_q \mod q$ \;
$S = S_q + q \cdot (i_q \cdot (S_p-S_q) \mod p)$ \tcp*[r]{Same as in line~\ref{alg-crt-rsa-unprotected-S} of Alg.~\ref{alg-crt-rsa-unprotected}}
\If{$(S-S'_p \not\equiv 0 \mod p)$ \Or $(S-S'_q \not\equiv 0 \mod q)$}{ \label{alg-crt-rsa-aumuller-test3}
	\Return $\mathsf{error}$ \;
}
\BlankLine
$S_{pt} = S'_p \mod t$ \;
$S_{qt} = S'_q \mod t$ \;
$d_{pt} = d'_p \mod (t-1)$ \;
$d_{qt} = d'_q \mod (t-1)$ \;
\eIf{$S_{pt}^{d_{qt}} \not\equiv S_{qt}^{d_{pt}} \mod t$}{ \label{alg-crt-rsa-aumuller-ElGamal}
	\Return $\mathsf{error}$ \;
}{
	\Return $S$ \; \label{alg-crt-rsa-aumuller-return_S}
}

\caption{Aumüller CRT-RSA}
\label{alg-crt-rsa-aumuller}
\end{algorithm}

Alg.~\ref{alg-crt-rsa-aumuller} does some verifications during the computations, and reports an error in case a fault injection can cause a malformed signature susceptible of unveiling $p$ and $q$.
More precisely, an error is returned in either of these seven cases:
\begin{enumerate}
\item $p'$ is not a multiple of $p$
(\emph{because this would amount to faulting $p$ in the unprotected algorithm})
\item $d'_p = d_p + \mathsf{random}_1 \cdot (p-1)$ is not equal to $d_p \mod (p-1)$
(\emph{because this would amount to faulting $d_p$ in the unprotected algorithm})
\item $q'$ is not a multiple of $q$
(\emph{because this would amount to faulting $q$ in the unprotected algorithm})
\item $d'_q = d_q + \mathsf{random}_2 \cdot (q-1)$ is not equal to $d_q \mod (q-1)$
(\emph{because this would amount to faulting $d_q$ in the unprotected algorithm})
\item $S-S'_p \mod p$ is nonzero
(\emph{because this would amount to faulting the recombination modulo $p$ in the unprotected algorithm})
\item $S-S'_q \mod q$ is nonzero
(\emph{because this would amount to faulting the recombination modulo $q$ in the unprotected algorithm})
\item $S_{pt}^{d_q} \mod t$ is not equal to $S_{qt}^{d_p} \mod t$
(\emph{this checks simultaneously for the integrity of $S'_p$ and $S'_q$})
\end{enumerate}

Notice that the last verification could not have been done on the unprotected algorithm,
it constitutes the added value of  Aumüller \etal's algorithm.
These seven cases are \emph{informally} assumed to protect the algorithm against the BellCoRe attack.
The criteria for fault detection is not to detect all faults;
for instance, a fault on the final return of $S$ (line~\ref{alg-crt-rsa-aumuller-return_S}) is not detected.
However, of course, such a fault is not exploitable by a BellCoRe attack.

\begin{remark}
Some parts of the Aumüller algorithm are actually not intended to protect against fault injection attacks,
but against side-channel analysis, such as the simple power analysis (SPA).
This is the case of lines~\ref{alg-crt-rsa-aumuller-line_dpa_p} and~\ref{alg-crt-rsa-aumuller-line_dpa_q} in Alg.~\ref{alg-crt-rsa-aumuller}.
These SPA attacks consist in monitoring via a side-channel the activity of the chip,
in a view to extract the secret exponent, using \emph{generic} methods described in~\cite{kocher-dpa_and_related_attacks}
or more \emph{accurate} techniques such as wavelet transforms~\cite{NIAT11_wavelet,YS:HASP12}.
They can be removed if a minimalist protection against only fault injection attacks is looked for;
but as they do not introduce weaknesses (in this very specific case), they are simply kept as such.
\end{remark}

\subsubsection{Vigilant}

The CRT-RSA algorithm of Vigilant~\cite{DBLP:conf/ches/Vigilant08} also considers computations in a larger ring than $\Z_p$ (abbreviation for $\Z/p\Z$) and $\Z_q$, to enable verifications.
In this case, a small random number $r$ is cast, and computations are carried out in $\Z_{p \cdot r^2}$ and $\Z_{q \cdot r^2}$.
In addition, the computations are now conducted not on the plain message $m$, but on an encoded message $m'$,
built using the CRT as the solution of those two requirements:
\begin{compactenum}
\item[\emph{i}: ]  $m' \equiv m \mod N$, and
\item[\emph{ii}: ] $m' \equiv 1+r \mod r^2$.
\end{compactenum}
This system of equations has a single solution modulo $N \times r^2$, because $N$ and $r^2$ are coprime.
Such a representation allows to conduct in parallel
the functional CRT-RSA (line \emph{i}) and
a verification (line \emph{ii}).
The verification is elegant, as it leverages this remarkable equality:
$(1+r)^{d_p} = \sum_{i=0}^{d_p} {d_p \choose i} \cdot r^i \equiv 1 + d_p \cdot r \mod r^2$.
Thus, as opposed to Aumüller \etal's CRT-RSA, which requires one exponentiation (line~\ref{alg-crt-rsa-aumuller-ElGamal} of Alg.~\ref{alg-crt-rsa-aumuller}), the verification of Vigilant's algorithm adds only one affine computation (namely $1 + d_p \mod r^2$).

The original description of Vigilant's algorithm involves some trivial computations on $p$ and $q$, such as $p-1$, $q-1$ and $p \times q$.
Those can be faulted, in such a way the BellCoRe attack becomes possible despite all the tests.
Thus, a patch by Coron \etal has been released in~\cite{DBLP:conf/fdtc/CoronGMPV10} to avoid the reuse of $\widehat{p-1}$, $\widehat{q-1}$ and $\widehat{p \times q}$ in the algorithm.

\subsubsection{Kim}

Kim, Kim, Han and Hong propose in~\cite{Kim:2011:ECA:2010601.2010865} a CRT-RSA algorithm that is based on a collaboration between a customized modular exponentiation and verifications at the recombination level based on Boolean operations.
The underlying protection concepts being radically different from the algorithms of Shamir, Aumüller and Vigilant, we choose not to detail this interesting countermeasure.

\subsubsection{Other Miscellaneous Fault Injections Attacks}

When the attacker has the power to focus its fault injections on \emph{specific bits} of \emph{sensitive resources},
then more challenging security issues arise~\cite{DBLP:series/isc/BerzatiCG12}.
These threats require a highly qualified expertise level, and are thus considered out of the scope of this paper.

Besides, for completeness, we mention that other fault injections mitigating techniques have been promoted,
such as the \emph{infective computation scheme} (refer to the seminal paper~\cite{DBLP:conf/ccs/BlomerOS03}).
This family of protections, although interesting, is neither covered by this article.

\bigskip

In this paper, we will focus on three implementations, namely
the unprotected one (Sec.~\ref{sec-unprotected}),
the one protected by Shamir's countermeasure (Sec.~\ref{sec-shamir}),
and the one protected by Aumüller \etal's countermeasure (Sec.~\ref{sec-aumuller}).

\section{Formal Methods}
\label{sec-approach}

For all the countermeasures presented in the previous section (Sec.~\ref{sec-bellcore}),
we can see that no formal proof of resistance against attacks is claimed.
Informal arguments are given, that convince that for some attack scenarii, the attack attempts are detected hence harmless.
Also, an analysis of the probability that an attack succeeds by chance (with a low probability of $1/t$) is carried out,
however, this analysis strong\-ly relies on assumptions on the faults distribution.
Last but not least, the algorithms include protections against both passive side-channel attacks (typically SPA) and active side-channel attacks, which makes it difficult to analyze for instance the minimal code to be added for the countermeasure to be correct.

\subsection{CRT-RSA and Fault Injections}

Our goal is to prove that a given countermeasure works,
\ie, that it delivers a result which does leak information about neither $p$ nor $q$ even when the implementation is subject to fault injections and to a BellCoRe attack.
In addition, we wish to reach this goal with the two following assumptions:

\begin{compactitem}
\item our proof applies to a very general attacker model, and
\item our proof applies to any implementation that is a (strict) refinement of the abstract algorithm.
\end{compactitem}
\smallskip

First, we must define what computation is done, and what is our threat model.

\begin{definition}[CRT-RSA]
The CRT-RSA computation takes as input a message $m$,
assumed known by the attacker, and a secret key $(p, q, d_p, d_q, i_q)$.
Then, the implementation is free to instantiate any variable, but must return a result equal to
$S = S_q + q \cdot (i_q \cdot (S_p - S_q) \mod p)$, where:
\begin{compactitem}
\item $S_p = m^{d_p} \mod p$, and
\item $S_q = m^{d_q} \mod q$.
\end{compactitem}
\label{def-crtrsa}
\end{definition}

\begin{definition}[fault injection]
An attacker is able to request RSA computations, as per Def.~\ref{def-crtrsa}.
During the computation, the attacker can modify any intermediate value by setting it to either a random value or zero.
At the end of the computation the attacker can read the result.
\label{def-faultinj}
\end{definition}
Of course, the attacker cannot read the intermediate values used during the computation,
since the secret key and potentially the modulus factors are used.
Such ``whitebox'' attack would be too powerful;
nonetheless, it is very hard in practice for an attacker to be able to access intermediate variables,
due to specific protections (\eg, blinding) and noise in the side-channel leakage (\eg, power consumption, electromagnetic emanation).
Remark that our model only takes into account fault injection on data;
the control flow is supposed not to be mutable.

As a side remark, we notice that the fault injection model of Def.~\ref{def-faultinj} corresponds to that of Vigilant (\cite{DBLP:conf/ches/Vigilant08}), with the exception that the conditional tests can also be faulted.
To summarize, an attacker can:
\begin{compactitem}
\item modify a value in memory (\emph{permanent fault}), and
\item modify a value in a local register, cache, or bus (\emph{transient fault}),
\end{compactitem}
but cannot
\begin{compactitem}
\item inject a permanent fault in the input data (message and secret key), nor
\item modify the algorithm control flow graph.
\end{compactitem}
\smallskip

The independence of the proofs on the algorithm implementation demands that the algorithm is described at a high level.
The two properties that characterize the relevant level are as follows:

\begin{compactenum}
\item The description should be low level enough for the attack to work if protections are not implemented.
\item Any additional intermediate variable that would appear during refinement could be the target of an attack,
but such a fault would propagate to an intermediate variable of the high level description, thereby having the same effect.
\end{compactenum}
From those requirements, we deduce that:
\begin{compactenum}
\item The RSA description must exhibit the computation modulo $p$ and $q$ and the CRT recombination;
typically, a completely blackbox description, where the computations would be realized in one go without intermediate variables, is not conceivable.
\item However, it can remain abstract, especially for the computational parts\footnote{For example, a fault in the implementation of the multiplication is either inoffensive, and we do not need to care about it, or it affects the result of the multiplication, and our model take it into account without going into the details of how the multiplication's is computed}.
\end{compactenum}
\smallskip

In our approach, the protections must thus be considered as an augmentation of the unprotected code,
\ie, a derived version of the code where additional variables are used.
The possibility of an attack on the unprotected code attests that the algorithm is described at the adequate level,
while the impossibility of an attack (to be proven) on the protected code shows that added protections are useful in terms of resistance to attacks.
\begin{remark}
The algorithm only exhibits evidence of safety.
If after a fault injection, the algorithm does not simplify to an error detection, then it might only reveal that some simplification is missing.
However, if it does not claim safety, it produces a \emph{simplified} occurrence of a possible weakness to be investigated further.
\end{remark}

\subsection{\finja}

Several tools are \emph{a priori} suitable for a formal analysis of CRT-RSA.
PARI/GP is a specialized computer algebra system, primarily aimed at solving number theory problems.
Although PARI/GP can do a fair amount of symbolic manipulation, it remains limited compared to systems like Axiom, Magma, Maple, Mathematica, Maxima, or Reduce.
Those last software also fall short to implement automatically number theoretic results like Euler's theorem.
This explains why we developed from scratch a system to reason on modular numbers from a formal point of view.
Our system is not general, in that it cannot for instance factorize terms in an expression.
However, it is able to simplify recursively what is simplifiable from a set of unambiguous rules.
This behavior is suitable to the problem of resistance to fault attacks,
because the redundancy that is added in the computation is meant to be simplified at the end (if no faults happened).

Our tool \finja works within the framework of modular arithmetic, which is the mathematical framework of CRT-RSA computations.
The general idea is to represent the computation term as a tree which encodes the computation properties.
This term can be simplified by \finja, using rules from arithmetic and the properties encoded in the tree.
Fault injections in the computation term are simulated by changing the properties of a subterm, thus impacting the simplification process.
An attack success condition is also given and used on the term resulting from the simplification to check whether the corresponding attack works on it.
The outputs of \finja are in HTML form: easily readable reports are produced, which contains all the information about the possible fault injections and their outcome.

\subsubsection{Computation Term}
\label{finja-term}

The computation is expressed in a convenient statement-based input language.
This language's Backus Normal Form (BNF) is given in Fig.~\ref{bnf}.

\begin{figure}
\begin{minipage}[c]{\linewidth}
\begin{minipage}[c]{0.05\linewidth}
~
\end{minipage}
\begin{minipage}[c]{0.85\linewidth} {\smaller 
\begin{Verbatim}[commandchars=@~&,numbers=left,numbersep=5pt,frame=lines]
term    ::= ( stmt )* 'return' mp_expr ';'
stmt    ::= ( decl | assign | verif ) ';'
decl    ::= 'noprop' mp_var ( ',' mp_var )*   @label~bnf:noprop&
          | 'prime' mp_var ( ',' mp_var )*    @label~bnf:prime&
assign  ::= var ':=' mp_expr                  @label~bnf:assign&
verif   ::= 'if' mp_cond 'abort with' mp_expr @label~bnf:verif&
mp_expr ::= '{' expr '}' | expr               @label~bnf:mpe&
expr    ::= '(' mp_expr ')'
          | '0' | '1' | var                   @label~bnf:expr1&
          | '-' mp_expr
          | mp_expr '+' mp_expr
          | mp_expr '-' mp_expr
          | mp_expr '*' mp_expr
          | mp_expr '^' mp_expr
          | mp_expr 'mod' mp_expr             @label~bnf:expr2&
mp_cond ::= '{' cond '}' | cond               @label~bnf:mpc&
cond    ::= '(' mp_cond ')'
          | mp_expr '=' mp_expr               @label~bnf:cond1&
          | mp_expr '!=' mp_expr
          | mp_expr '=[' mp_expr ']' mp_expr  @label~bnf:eqmod&
          | mp_expr '!=[' mp_expr ']' mp_expr @label~bnf:neqmod&
          | mp_cond '/\' mp_cond
          | mp_cond '\/' mp_cond              @label~bnf:cond2&
mp_var  ::= '{' var '}' | var
var     ::= [a-zA-Z][a-zA-Z0-9_']*            @label~bnf:var&
\end{Verbatim}
} 
\caption{\label{bnf} BNF of \finja's input language.}
\end{minipage}
\end{minipage}
\end{figure}

A computation term is defined by a list of statements finished by a {\tt return} statement.
Each statement can either:
\begin{itemize}
\item declare a variable with no properties (line~\ref{bnf:noprop});
\item declare a variable which is a prime number (line~\ref{bnf:prime});
\item declare a variable by assigning it a value (line~\ref{bnf:assign}), in this case the properties of the variable are the properties of the assigned expression;
\item perform a verification (line~\ref{bnf:verif}).
\end{itemize}

As can be seen in lines~\ref{bnf:expr1} to~\ref{bnf:expr2}, an expression can be:
\begin{itemize}
\item zero, one, or an already declared variable;
\item the sum (or difference) of two expressions;
\item the product of two expressions;
\item the exponentiation of an expression by another;
\item the modulus of an expression by another.
\end{itemize}

The condition in a verification can be (lines~\ref{bnf:cond1} to~\ref{bnf:cond2}):
\begin{itemize}
\item the equality or inequality of two expressions;
\item the equivalence or non-equivalence of two expressions modulo another (lines~\ref{bnf:eqmod} and~\ref{bnf:neqmod});
\item the conjunction or disjunction of two conditions.
\end{itemize}

Optionally, variables (when declared using the {\tt prime} or {\tt noprop} keywords), expressions, and conditions can be protected (lines~\ref{bnf:noprop}, \ref{bnf:prime}, \ref{bnf:mpe} and~\ref{bnf:mpc}, {\tt mp} stands for ``maybe protected'') from fault injection by surrounding them with curly braces.
This is useful for instance when it is needed to express the properties of a variable which cannot be faulted in the studied attack model.
For example, in CRT-RSA, the definitions of variables $d_p$, $d_q$, and $i_q$ are protected because they are seen as input of the computation.

Finally, line~\ref{bnf:var} gives the regular expression that variable names must match (they start with a letter and then can contain letters, numbers, underscore, and simple quote).

\medskip

After it is read by \finja, the computation expressed in this input language is transformed into a tree (just like the abstract syntax tree in a compiler).
This tree encodes the arithmetical properties of each of the intermediate variable, and thus its dependencies on previous variables.
The properties of intermediate variables can be everything that is expressible in the input language.
For instance, being null or being the product of other terms (and thus, being a multiple of each of them), are possible properties.

\subsubsection{Fault Injection}
\label{finja-fi}

A fault injection on an intermediate variable is represented by changing the properties of the subterm (a node and its whole subtree in the tree representing the computation term) that represent it.
In the case of a fault which forces at zero, then the whole subterm is replaced by a term which only has the property of being null.
In the case of a randomizing fault, by a term which have no properties.

\finja simulates \emph{all the possible fault injections} of the attack model it is launched with.
The parameters allow to choose:
\begin{itemize}
\item \emph{how many faults} have to be injected (however, the number of tests to be done is impacted by a factorial growth with this parameter, as is the time needed to finish the computation of the proof);
\item \emph{the type} of each fault (\emph{randomizing} or \emph{zeroing});
\item if \emph{transient} faults are possible or if only \emph{permanent} faults should be performed.
\end{itemize}

\subsubsection{Attack Success Condition}
\label{finja-asc}

The attack success condition is expressed using the same condition language as presented in Sec.~\ref{finja-term}.
It can use any variable introduced in the computation term, plus two special variables {\tt \_} and {\tt @} which are respectively bound to the expression returned by the computation term as given by the user and to the expression returned by the computation with the fault injections.
This success condition is checked for each possible faulted computation term.

\subsubsection{Simplification Process}
\label{finja-reduce}

The simplification is implemented as a recursive traversal of the term tree, based on pattern-matching.
It works just like a naive interpreter would, except it does symbolic computation only, and reduces the term based on rules from arithmetic.
Simplifications are carried out in $\Z$ ring, and its $\Z_N$ subrings.
The tool knows how to deal with most of the $\Z$ ring axioms:
\begin{itemize}
\item the neutral elements ($0$ for sums, $1$ for products);
\item the absorbing element ($0$, for products);
\item inverses and opposites (only if $N$ is prime);
\item associativity and commutativity.
\end{itemize}
However, it does not implement distributivity as it is not confluent. Associativity is implemented by flattening as much as possible (``removing'' all unnecessary parentheses), and commutativity is implemented by applying a stable sorting algorithm on the terms of products or sums.

The tool also knows about most of the properties that are particular to $\Z_N$ rings and applies them when simplifying a term modulo $N$:
\begin{itemize}
\item identity:
  \begin{itemize}
  \item $(a \mod N) \mod N = a \mod N$,
  \item $N^k \mod N = 0$;
  \end{itemize}
\item inverse:
  \begin{itemize}
  \item $(a \mod N) \times (a^{-1} \mod N) \mod N = 1$,
  \item $(a \mod N) + (-a \mod N) \mod N = 0$;
  \end{itemize}
\item associativity and commutativity:
  \begin{itemize}
  \item $(b \mod N) + (a \mod N) \mod N = a + b \mod N$,
  \item $(a \mod N) \times (b \mod N) \mod N = a \times b \mod N$;
  \end{itemize}
\item subrings: $(a \mod N \times m) \mod N = a \mod N$.
\end{itemize}

In addition to those properties a few theorems are implemented to manage more complicated cases where the properties are not enough when conducting symbolic computations:
\begin{itemize}
\item Fermat's little theorem;
\item its generalization, Euler's theorem.
\end{itemize}

\paragraph{Example.}
If we have the term {\tt t := a + b * c}, it can be faulted in five different ways (using the randomizing fault):
\begin{enumerate}
\item {\tt t := Random}, the final result is faulted;
\item {\tt t := Random + b * c}, $a$ is faulted;
\item {\tt t := a + Random}, the result of $b \times c$ is faulted;
\item {\tt t := a + Random * c}, $b$ is faulted;
\item {\tt t := a + b * Random}, $c$ is faulted.
\end{enumerate}
If the properties that interest us is to know whether $t$ is congruent with $a$ modulo $b$, we can use {\tt t =[b] a} as the attack success condition.
Of course it will be true for $t$, but it will only be true for the fifth version of faulted $t$.
If we had used the zeroing fault, it would also have been true for the third and fourth versions.

\section{Study of an Unprotected CRT-RSA Computation}
\label{sec-unprotected}

The description of the unprotected CRT-RSA computation in \finja code is given in Fig.~\ref{code-unprotec} (note the similarity of \finja's input code with Alg.~\ref{alg-crt-rsa-unprotected}).

\begin{figure}[h]
\begin{minipage}[c]{\linewidth}
\begin{minipage}[c]{0.05\linewidth}
~
\end{minipage}
\begin{minipage}[c]{0.85\linewidth} {\smaller 
\begin{Verbatim}[commandchars=|~&,numbers=left,numbersep=5pt,frame=lines]
noprop m, e ;
prime  {p}, {q} ;

dp := { e^-1 mod (p-1) } ;
dq := { e^-1 mod (q-1) } ;
iq := { q^-1 mod p } ;

Sp := m^dp mod p ;
Sq := m^dq mod q ;

S := Sq + (q * (iq * (Sp - Sq) mod p)) ;

return S ;

%%

_ != @ /\ ( _ =[p] @ \/ _ =[q] @ )
\end{Verbatim}
} 
\end{minipage}
\end{minipage}
\caption{\label{code-unprotec} \finja code for the unprotected CRT-RSA computation.}
\end{figure}

As we can see, the definitions of $d_p$, $d_q$, and $i_q$ are protected so the computation of the values of these variables cannot be faulted (since they are seen as inputs of the algorithm).
After that, $S_p$ and $S_q$ are computed and then recombined in the last expression, as in Def.~\ref{def-crtrsa}.

To test whether the BellCoRe attack works on a faulted version $\widehat{S}$, we perform the following tests (we note $|S|$ for the simplified version of $S$):

\begin{compactenum}
\item is $|S|$ equal to $|\widehat{S}|$?
\item is $|S \mod p|$ equal to $|\widehat{S} \mod p|$?
\item is $|S \mod q|$ equal to $|\widehat{S} \mod q|$?
\end{compactenum}

If the first test is false and at least one of the second and third is true, we have a BellCoRe attack, as seen in Sec.~\ref{sec-bellcore}.
This is what is described in the attack success condition (after the {\tt \%\%} line).

Without transient faults enabled, and in a single fault model, there are $12$ different fault injections of which $8$ enable a BellCoRe attack with a randomizing fault, and $9$ with a zeroing fault.
As an example, replacing the intermediate variable holding the value of $i_q \cdot (S_p-S_q) \mod p$ in the final expression with zero or a random value makes the first and second tests false, and the last one true, and thus allows a BellCoRe attack.

\section{Study of Shamir's Countermeasure}
\label{sec-shamir}

The description, using \finja's formalism, of the CRT-RSA computation allegedly protected by Shamir's countermeasure is given in Fig.~\ref{code-shamir} (again, note the similarity with Alg.~\ref{alg-crt-rsa-shamir}).

\begin{figure}[h]
\begin{minipage}[c]{\linewidth}
\begin{minipage}[c]{0.05\linewidth}
~
\end{minipage}
\begin{minipage}[c]{0.85\linewidth} {\smaller 
\begin{Verbatim}[commandchars=|~&,numbers=left,numbersep=5pt,frame=lines]
noprop error, m, d ;
prime {p}, {q}, r ;
iq := { q^-1 mod p } ;

p' := p * r ;
dp := d mod ((p-1) * (r-1)) ;
Sp' := m^dp mod p' ;

q' := q * r ;
dq := d mod ((q-1) * (r-1)) ;
Sq' := m^dq mod q' ;

Sp := Sp' mod p ; |label~code_shamir_sp&
Sq := Sq' mod q ;

S := Sq + (q * (iq * (Sp - Sq) mod p)) ; |label~code_shamir_recomb&

if Sp' !=[r] Sq' abort with error ; |label~code_shamir_test&

return S ;

%%

_ != @ /\ ( _ =[p] @ \/ _ =[q] @ )
\end{Verbatim}
} 
\end{minipage}
\end{minipage}
\caption{\label{code-shamir} \finja code for the Shamir CRT-RSA computation.}
\end{figure}

Using the same settings as for the unprotected implementation of CRT-RSA, we find that among the $31$ different fault injections, $10$ enable a BellCoRe attack with a randomizing fault, and $9$ with a zeroing fault.
This is not really surprising, as the test which is done on line~\ref{code_shamir_test} does not verify if a fault is injected during the computations of $S_p$ or $S_q$, nor during their recombination in $S$.
For instance zeroing or randomizing the intermediate variable holding the result of $S_p - S_q$ during the computation of $S$ (line~\ref{code_shamir_recomb}) results in a BellCoRe attack.
To explain why there is this problem in Shamir's countermeasure,
some context might be necessary.
It can be noted that the fault to inject in the countermeasure must be more accurate in timing (since it targets an intermediate variable obtained by a \emph{subtraction}) than the faults to achieve a BellCoRe attack on the unprotected CRT-RSA (since a fault during an \emph{exponentiation} suffices).
However, there is today a consensus to believe that it is very easy to pinpoint in time any single operation of a CRT-RSA algorithm,
using a simple power analysis method~\cite{kocher-dpa_and_related_attacks}.
Besides, timely fault injection benches exist.
Therefore, the weaknesses in Shamir's countermeasure can indeed be practically exploited.

If the attacker can do \emph{transient faults}, there are a lot more attacks:
$66$ different possible fault injections of which $24$ enable a BellCoRe attack with a randomizing fault and $22$ with a zeroing fault.
In practice, a transient faults would translate into faulting the variable when it is read (\eg, in a register or on a bus), rather than in (persistent) memory.
This behavior could also be the effect of a fault injection in cache, which is later replaced with the good value when it is read from memory again.
To the authors knowledge, these are not impossible situations.
Nonetheless, growing the power of the attacker to take that into account break some very important assumptions that are classical (sometimes even implicit) in the literature.
It does not matter that the parts of the secret key are stored in a secure ``key container'' if their values can be a faulted at read time.
Indeed, we just saw that allowing this kind of fault enable even more possibilities to carry out a BellCoRe attack successfully on a CRT-RSA computation protected by the Shamir's countermeasure.
For instance, if the value of $p$ is randomized for the computation of the value of $S_p$ (line~\ref{code_shamir_sp}),
then we have $S \neq \widehat{S}$, but also $S \equiv \widehat{S} \mod q$, which enables a BellCoRe attack, as seen in Sec.~\ref{sec-bellcore}.

It is often asserted that the countermeasure of Shamir is unpractical due to its need for $d$
(as mentioned in~\cite{DBLP:conf/ches/AumullerBFHS02} and~\cite{DBLP:conf/ches/Vigilant08}),
and because there is a possible fault attack on the recombination,
\ie, line~\ref{code_shamir_recomb}
(as mentioned in~\cite{DBLP:conf/ches/Vigilant08}).
However, the attack on the recombination can be checked easily, by testing that $S-S_p \not\equiv 0 \mod p$ and $S-S_q \not\equiv 0 \mod q$ before returning the result.
Notwithstanding, to our best knowledge, it is difficult to detect all the attacks our tool found, and so the existence of these attacks (new, in the sense they have not all been described previously) is a compelling reason for not implementing Shamir's CRT-RSA.

\section{Study of Aumüller \etal's Countermeasure}
\label{sec-aumuller}

The description of the CRT-RSA computation protected by Aumüller \etal's countermeasure is given in Fig.~\ref{code-aumuller} (here too, note the similarity with Alg.~\ref{alg-crt-rsa-aumuller})

\begin{figure}[h]
\begin{minipage}[c]{\linewidth}
\begin{minipage}[c]{0.05\linewidth}
~
\end{minipage}
\begin{minipage}[c]{0.85\linewidth} {\smaller 
\begin{Verbatim}[commandchars=|~&,numbers=left,numbersep=5pt,frame=lines]
noprop error, m, e, r1, r2 ;
prime {p}, {q}, t ;

dp := { e^-1 mod (p-1) } ;
dq := { e^-1 mod (q-1) } ;
iq := { q^-1 mod p } ;

p' := p * t ;
dp' := dp + r1 * (p-1) ;
Sp' := m^dp' mod p' ;

if p' !=[p] 0 \/ dp' !=[p-1] dp abort with error ; |label~code_aumuller_1&

q' := q * t ;
dq' := dq + r2 * (q-1) ;
Sq' := m^dq' mod q' ;

if q' !=[q] 0 \/ dq' !=[q-1] dq abort with error ; |label~code_aumuller_2&

Sp := Sp' mod p ;
Sq := Sq' mod q ;

S := Sq + (q * (iq * (Sp - Sq) mod p)) ;

if S !=[p] Sp' \/ S !=[q] Sq' abort with error ;  |label~code_aumuller_3&

Spt := Sp' mod t ;
Sqt := Sq' mod t ;
dpt := dp' mod (t-1) ;
dqt := dq' mod (t-1) ;

if Spt^dqt !=[t] Sqt^dpt abort with error ;  |label~code_aumuller_eg&

return S ;

%%

_ != @ /\ ( _ =[p] @ \/ _ =[q] @ )
\end{Verbatim}
} 
\end{minipage}
\end{minipage}
\caption{\label{code-aumuller} \finja code for the Aumüller \etal CRT-RSA computation.}
\end{figure}

Using the same method as before, we can prove that on the $52$ different possible faults, \emph{none} of which allow a BellCoRe attack, whether the fault is zero or random.
This is a proof that the Aumüller \etal's countermeasure works when there is one fault\footnote{This result is worthwhile some emphasis:
the genuine algorithm of Aumüller is thus \emph{proved} resistant against single-fault attacks.
At the opposite, the CRT-RSA algorithm of Vigilant is not immune to single fault attacks (refer to~\cite{DBLP:conf/fdtc/CoronGMPV10}), and the corrections suggested in the same paper by Coron \etal have not been proved yet.}.

Since it allowed more attacks on the Shamir's countermeasure, we also tested the Aumüller \etal's countermeasure against \emph{transient faults} such as described in Sec.~\ref{sec-shamir}.
There are $120$ different possible fault injections when transient faults are activated, and Aumüller \etal's countermeasure is resistant against such fault injections too.

We also used \finja to confirm that the computation of $d_p$, $d_q$, and $i_q$ (in terms of $p$, $q$, and $d$) must not be part of the algorithm.
The countermeasure effectively needs these three variables to be inputs of the algorithm to work properly.
For instance there is a BellCoRe attack if $d_q$ happens to be zeroed.
However, even with $d_p$, $d_q$, and $i_q$ as inputs, we can still attempt to attack a CRT-RSA implementation protected by the Aumüller \etal's countermeasure by doing more than one fault.

We then used \finja to verify whether Aumüller \etal's countermeasure would be resistant against \emph{high order} attacks, starting with two faults.
We were able to break it if at least one of the two faults was a zeroing fault.
We found that this zeroing fault was used to falsify the condition of a verification, which is possible in our threat-model, but which was not in the one of the authors of the countermeasure.
If we protect the conditions against fault injection, then the computation is immune two double-fault attacks too.
However, even in this less powerful threat-model, a CRT-RSA computation protected by Aumüller \etal's countermeasure is breakable using $3$ faults, two of which must be zeroing the computations of $d_{pt}$ and $d_{qt}$.

\section{Conclusions and Perspectives}
\label{sec-perspectives}

We have formally proven the resistance of the Aumüller \etal's countermeasure against the BellCoRe attack by fault injection on CRT-RSA.
To our knowledge, it is the first time that a formal proof of security is done for a BellCoRe countermeasure.

During our research, we have raised several questions about the assumptions traditionally made by countermeasures.
The possibility of fault at read time is, in particular, responsible for many vulnerabilities.
The possibility of such fault means that part of the secret key can be faulted (even if only for one computation).
It allows an interesting BellCoRe attack on a computation of CRT-RSA protected by Shamir's countermeasure.
We also saw that the assumption that the result of conditional expression cannot be faulted, which is widespread in the literature, is a dangerous one as it increased the number of fault necessary to break Aumüller \etal's countermeasure from $2$ to $3$.

The first of these two points demonstrates the lack of formal studies of fault injection attack and their countermeasures,
while the second one shows the importance of formal methods in the field of implementation security.

\medskip

As a first perspective, we would like to address the hardening of software codes of CRT-RSA under the threat of a bug attack.
This attack has been introduced by Biham, Carmeli and Shamir~\cite{DBLP:conf/crypto/BihamCS08} at CRYPTO 2008.
It assumes that a hardware has been trapped in such a way that there exists two integers $a$ and $b$,
for which the multiplication is incorrect.
In this situation, Biham, Carmeli and Shamir mount an explicit attack scenario where the knowledge of $a$ and $b$ is leveraged to produce a faulted result, that can lead to a remote BellCoRe attack.
For sure, testing for the correct functionality of the multiplication operation is impractical
(it would amount to an exhaustive verification of $2^{128}$ multiplications on $64$~bit computer architectures).
Thus, it can be imagined to use a countermeasure, like that of Aumüller, to detect a fault (caused logically).
Our aim would be to assess in which respect our fault analysis formal framework allows to validate the security of the protection.
Indeed, a fundamental difference is that the fault is not necessarily injected at \emph{one} random place,
but can potentially show up at {\emph{several} places.

As another perspective, we would like to handle the repaired countermeasure of Vigilant~\cite{DBLP:conf/fdtc/CoronGMPV10} and the countermeasure of Kim~\cite{Kim:2011:ECA:2010601.2010865}.
Regarding Vigilant, the difficulty that our verification framework in OCaml~\cite{OCaml} shall overcome is to decide how to inject the remarkable identity $(1+r)^{d_p} \equiv 1 + d_p \cdot r \mod r^2$:
either it is kept as such such, like an \emph{ad hoc} theorem (but we need to make sure it is called only at relevant places, since it is not confluent), or it is made more general (but we must ascertain that the verification remains tractable).
However, this effort is worthwhile%
\footnote{Some results will appear in the proceedings of the 3rd ACM SIGPLAN Program Protection and Reverse Engineering Workshop (PPREW 2014)~\cite{PR:PPREW14}, collocated with POPL 2014.},
because the authors themselves say in the conclusion of their article~\cite{DBLP:conf/fdtc/CoronGMPV10} that:
\begin{quote}
\hspace*{-2mm}\emph{``Formal proof of the FA-resistance of Vigilant's sche\-me including our countermeasures is still an open (and challenging) issue.''}
\end{quote}
Regarding the CRT-RSA algorithm from Kim, the computation is very detailed (it goes down to the multiplication level), and involves Boolean operations (\texttt{and}, \texttt{xor}, \etc). To manage that, more expertise about both arithmetic and logic must be added to our software.

Eventually, we wish to answer a question raised by Vigilant~\cite{DBLP:conf/ches/Vigilant08} about the prime $t$ involved in Aumüller \etal's countermeasure:
\begin{quote}
\emph{``Is it fixed or picked at random in a fixed table?''}
\end{quote}
The underlying issue is that of \emph{replay} attacks on CRT-RSA, that are more complicated to handle;
indeed, they would require a formal system such as ProVerif~\cite{ProVerif}, that is able to prove interactive protocols.

\medskip

Concerning the tools we developed during our research, they currently only allow to study fault injection in the data, and not in the control flow,
it would be interesting to enable formal study of fault injections affecting the control flow.

\medskip

Eventually, we would like to define and then implement an automatic code mutation algorithm that could transform an un\-pro\-tect\-ed CRT-RSA into a protected one.
We know that with a few alterations (see that the differences between Alg.~\ref{alg-crt-rsa-unprotected} and Alg.~\ref{alg-crt-rsa-aumuller} are enumerable),
this is possible.
Such promising approach, if successful, would uncover the \emph{smallest possible} countermeasure of CRT-RSA against fault injection attacks.

\section*{Acknowledgements}

The authors wish to thank Jean-Pierre Seifert and Wieland Fischer for insightful comments and pieces of advice.
We are also grateful to the anonymous reviewers of \href{http://www.proofs-workshop.org/}{PROOFS 2013} (UCSB, USA),
who helped improve the preliminary version of this paper.
Eventually, we acknowledge precious suggestions contributed by Jean-Luc Danger, Jean Goubault-Larrecq, and Karine Heydemann.

\balance
\bibliographystyle{alpha}
\bibliography{sca}

\end{document}